\documentclass[a4paper]{article}
\usepackage{fullpage}
\usepackage{amsmath,amsthm,amssymb,amsfonts,mathrsfs,amsthm}
\usepackage{dsfont,yfonts,calligra}
\usepackage{graphicx,hyperref}
\usepackage{pxfonts,upgreek}
\usepackage{lmodern}
\usepackage{xcolor}
\usepackage{authblk}
\usepackage{framed}
\newenvironment{protocol}{\vspace{-1em}\begin{framed}
}{
\vspace{-3ex}\end{framed}\vspace{-1em}}

\DeclareMathAlphabet{\mathpzc}{OT1}{pzc}{m}{it}
\DeclareMathAlphabet{\mathcalligra}{T1}{calligra}{m}{n}

\def\Tr{\operatorname{Tr}} \def\>{\rangle} \def\<{\langle}
 \def\id{\mathsf{id}}

\def\mN{\mathcal{N}}

\def\mS{\mathcal{S}}
\def\sH{\mathcal{H}} 
\def\sS{{\boldsymbol{\mathsf{S}}}}  
\def\bound{\boldsymbol{\mathsf{L}}}

 \def\dec#1{\mathscr{D}}
 \def\openone{\mathds{1}}

\def\suff{\succ}

\def\mD{\mathcal{D}}

\def\Span{{\operatorname{span}}}
\def\mM{\mathcal{M}}
\def\supseteqs{\supseteq_{\operatorname{s}}}
\def\supseteqw{\supseteq_{\operatorname{w}}}
\def\mNenv{\mN_{\operatorname{env}}}

\newcommand{\set}[1]{\mathscr{#1}}
\newcommand{\ens}[1]{\mathcal{#1}}
\newcommand{\model}[1]{\mathfrak{#1}}
\newcommand{\povm}[1]{\mathbb{#1}}

\renewcommand{\qedsymbol}{\nobreak \ifvmode \relax \else
  \ifdim \lastskip<1.5em \hskip-\lastskip \hskip1.5em plus0em
  minus0.5em \fi \nobreak \vrule height0.75em width0.5em
  depth0.25em\fi}

\renewcommand{\ge}{\geqslant}
\renewcommand{\le}{\leqslant}

\newtheorem{theorem}{Theorem}
\newtheorem{corollary}{Corollary}
\newtheorem{lemma}{Lemma}

\newtheorem{definition}{Definition}

\theoremstyle{remark}
\newtheorem{remark}{Remark}

\theoremstyle{definition}

\newcommand{\bea}{\begin{eqnarray}}
\newcommand{\eea}{\end{eqnarray}}
\newcommand{\be}{\begin{equation}}
\newcommand{\ee}{\end{equation}}
\def\reff#1{(\ref{#1})}

\begin{document}



\title{Game-theoretic characterization of antidegradable channels}






\author[1]{Francesco Buscemi\thanks{buscemi@is.nagoya-u.ac.jp}}
\author[2]{Nilanjana Datta\thanks{n.datta@statslab.cam.ac.uk}}
\author[3]{Sergii Strelchuk\thanks{ss870@cam.ac.uk}}

\affil[1]{\footnotesize Graduate School of Information Science, Nagoya University, Chikusa-ku, Nagoya 464-8601, Japan\normalsize}
\affil[2]{\footnotesize Statistical Laboratory, University of Cambridge, Cambridge CB3 0WB, U.K.\normalsize}
\affil[3]{\footnotesize Department of Applied Mathematics and Theoretical Physics,
	University of Cambridge, Cambridge CB3 0WA, U.K.\normalsize}

\date{}

\maketitle

\begin{abstract}
We introduce a guessing game involving a quantum channel, three parties---the sender, the receiver and an eavesdropper, Eve---and a quantum public side channel. We prove that a necessary and sufficient condition for the quantum channel to be antidegradable, is that Eve wins the game. We thus obtain a complete operational characterization of antidegradable channels in a game-theoretic framework.
\end{abstract}





\section{Introduction}
There are numerous tasks in quantum information theory which involve the use of quantum channels.
A quantum channel has many different capacities, 
depending on the task at hand, the nature of the information transmitted, and available resources.
For example, the quantum capacity of a channel quantifies its potential for communication
of quantum information, whereas its private capacity quantifies its potential for secure communication of classical information~\cite{devetak}. Deciding whether a given quantum channel has a positive capacity is a non-trivial problem,
e.g. there does not exist a unique criterion to determine whether the quantum capacity of a given channel is zero. 
Classical channels with zero capacity are uninteresting in the information-theoretic sense. In contrast, quantum channels with zero capacity 
exhibit intriguing behavior as shown by the superactivation phenomenon~\cite{smith_quantum_2008}: there exist examples of pairs of channels with zero quantum capacity, 
which, when used in tandem, allow transmission of quantum information. 
One particular class of zero-capacity channels consists of {\em{antidegradable channels}}. For such a channel, a post-processing of its
environment can simulate the output of the channel~\cite{footnote1}. The no-cloning theorem~\cite{wootters_single_1982} ensures that such channels have zero quantum capacity. 
The simplest example of the latter is a 50\% erasure channel which with equal probability  either transmits the input state 
perfectly or replaces it with an erasure flag. However, there are other non-trivial examples of channels with zero quantum capacity,
e.g.,~the positive partial transpose (PPT) channels~\cite{horodecki_binding_1999}. In addition, antidegradable channels also have 
zero private capacity (unlike PPT channels), but whether they are the only non-trivial quantum channels with this property
is an open question (since there exist echo-correctable channels with arbitrarily small, but non-zero, private capacity~\cite{smith_extensive_2009,li_private_2009}). Therefore, the knowledge that a given channel has zero quantum and private capacity 
is not sufficient to conclude that it is antidegradable. This leads us to the following question: 
\medskip

\indent
(Q): Is there a setting in which one can obtain a complete operational characterization of antidegradable channels? 
\smallskip

In this paper we answer this question in the affirmative by constructing a game-theoretic framework which involves
the noisy quantum channel ${\cal{N}}$ (which we wish to characterize), a quantum public side channel ${\cal{S}}$,
and three parties: Alice (the sender), Bob (the receiver) and Eve (the eavesdropper).
Alice sends classical information to Bob through 
${\cal{N}}$, whose environment is accessible to Eve. Alice also sends information through  ${\cal{S}}$, which is accessible to
both Bob and Eve. 
\newpage

The game is constructed as follows (formal definitions are given in Section~\ref{sec:main}). 
\bigskip

\begin{protocol}


\subsection{The guessing game}\label{game}
\begin{enumerate}
\item  Alice chooses a letter $x$ at random from a given finite alphabet $\set{X}$, and encodes it in a bipartite state, say $\rho^x_{AA_0}$.
	\item The $A$ part of the input is sent through $\cal{N}$, while the $A_0$ part is transmitted via $\cal{S}$.
	\item Bob then obtains the output of $\mN$ while Eve receives the
	information that is transmitted to the channel's environment. In other words, she receives the output of the 
	{\em{complementary channel}} $\mNenv$ (see Section 2 for its definition). In addition, they both receive the output of $\cal{S}$.
	\item The task now, for both Bob and Eve, is to guess which letter $x$ Alice chose. Since Bob and Eve are competing, they both 
adopt the optimal guessing strategy they have available. Correspondingly, the reliabilities of their guesses is measured by the optimal 
guessing probabilities of the ensembles of states they receive.
	\item Bob wins the game whenever his guessing probability is \textit{strictly higher} than that of Eve (i.e. in the case of a draw, Eve wins).
\end{enumerate}
\end{protocol}
\bigskip

The situation is depicted in Figure~\ref{fig:1} below.
\begin{figure}[h]
	\begin{center}
		\includegraphics[width=8cm]{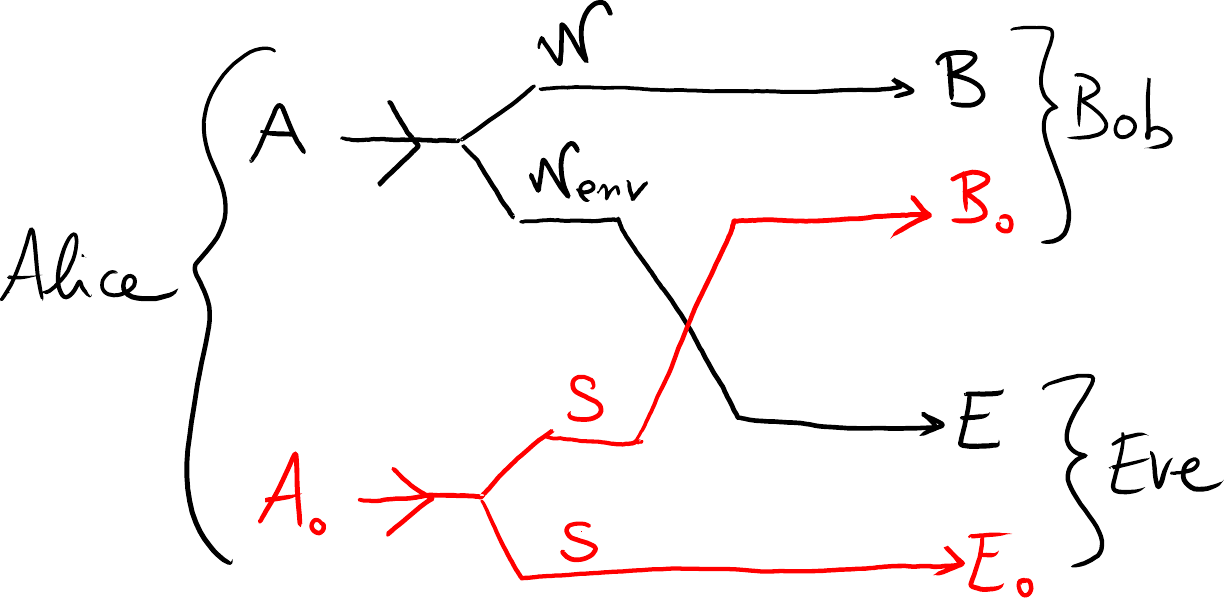}
	\end{center}
	\caption{Structure of the guessing game: Alice communicates with Bob using the quantum channel $\mN$ (i.e. the one which we want to characterize) and a quantum channel $\mS$, which is public, in the sense that it conveys the same output to Bob and Eve. 
A natural example of such a public channel is a symmetric channel~\cite{SSW2008,BO2012,BO2013}. Bob plays the guessing game against Eve, 
who has access to the environment of $\mN$ (labelled by $\mNenv$) and $\mS$.}
	\label{fig:1}
\end{figure}

To state our main result (Theorem~\ref{theo:main}) which leads to the characterization of antidegradable channels, we first 
introduce the notion of {\it{extension}} of a quantum channel: 
for any pair of quantum channels $({\cal{N}}_\alpha, {\cal N}_\beta)$ we say that ${\cal N}_\alpha$ is an {\it extension} of ${\cal N}_\beta$ if  
${\cal N}_\beta={\cal D}\circ{\cal N}_\alpha$ for some quantum channel ${\cal D}$. This is a generalization of the notion of {\it degradable extension} of the channel (introduced in~\cite{SmithSmolin08}) which corresponds to the case in which the channel ${\cal N}_\alpha$ is degradable and 
${\cal N}_\beta$ is complementary to it.  
Then our result can be stated as follows: for any given input ensemble of states, the guessing probability of the output ensemble of $\mN_\alpha$ is higher than that of $\mN_\beta$, if and only if $\mN_\alpha$ is an extension of $\mN_\beta$. We establish the above result by first proving its analogue for statistical comparison of bipartite states and then using Choi isomorphism.

Consider the case in which ${\cal N}_\beta$ is the quantum channel ${\cal N}$ employed in the guessing game~\ref{game}, and ${\cal N}_\alpha$ is the channel $\mNenv$
which is complementary to it. For this choice, our result (Theorem~\ref{theo:main})  implies that ${\cal N}$ is antidegradable \textit{if and only if} Eve always wins, 
regardless of the choice of Alice's encoding strategy. In other words, our result shows that, \textit{for any} channel which is not 
antidegradable, there exists (at least) one encoding strategy which Alice can choose to make Bob win the guessing game.

We note that even though the scenario of our guessing game is `cryptographic' in its nature (since Bob and Eve compete), 
proving that Bob is able to win against Eve in the guessing game is insufficient to conclude that any \textit{information-theoretic secrecy} can be established between Alice and Bob. This is because, in the guessing game, we only compare the 
guessing probabilities of Bob and Eve, and not the mutual informations between
the random variables corresponding to their respective inferences and 
that of the random variable corresponding to Alice's input. 
However, the game-theoretic scenario~\ref{game} has the particular advantage of singling out antidegradable channels as the only channels for which Eve necessarily wins. 
In other words, a {\em{necessary and sufficient condition}} for the quantum channel $\mN$ to be antidegradable is that Eve wins the guessing game, for any possible encoding strategy Alice may choose.

The paper is organized as follows. In Section~\ref{sec:main} we introduce the necessary notation and definitions, and then state our main result (Theorem~\ref{theo:main}). In Section~\ref{sec:states} we derive an analogue of Theorem~\ref{theo:main} for partial orderings of bipartite quantum states. In Section~\ref{sec:back-to-channels} we use this result, in conjunction with the Choi isomorphism, to obtain a proof of 
Theorem~\ref{theo:main}. Some further implications of Theorem~\ref{theo:main} 
for convex combinations of channels and their extensions are given in \ref{sec:further}. We end with a brief summary and open questions in Section~\ref{sec:conclude}.

\section{Main result}\label{sec:main}

\subsection{Notation and definitions}

In what follows, we only consider quantum systems defined on finite
dimensional Hilbert spaces $\sH$. We denote by $\bound(\sH)$ the set
of all linear operators
acting on $\sH$, and by $\sS(\sH)$ the set of all density operators (or
\textit{states}) $\rho\in\bound(\sH)$, with $\rho\ge 0$ and $\Tr[\rho]=1$. The
identity operator in $\bound(\sH)$ will be denoted by the symbol
$\openone$, whereas the identity map from $\bound(\sH)$ to
itself will be denoted by $\id$. A \textit{positive-operator valued measure} (POVM) is a family $\povm{P}=\{P^x\}_{x\in\set{X}}\subset\bound(\sH)$ of operators $P^x\ge 0$, labelled by a finite index set $\set{X}=\{x\}$ (i.e. the \textit{outcome set}), such that $\sum_{x\in\set{X}}P^x=\openone$.

In what follows, a
\emph{channel} is considered as a triple $(\sH_A,\sH_B,\mN)$, where $\sH_A$ is the input Hilbert space, $\sH_B$ is the output Hilbert space, and $\mN:\bound(\sH_A)\to \bound(\sH_B)$ is a completely positive, trace-preserving (CPTP) linear map. Where possible, we will denote a channel $(\sH_A,\sH_B,\mN)$ simply as $\mN$. The adjoint (Heisenberg dual) of a linear map $\mN:\bound(\sH_A)\to
\bound(\sH_B)$ is the linear map $\mN^*:\bound(\sH_B)\to
\bound(\sH_A)$ defined with respect to the Hilbert-Schmidt inner product by  $\Tr[\mN^*(X)\
Y]:=\Tr[X\ \mN(Y)]$, for all $X\in\bound(\sH_B)$ and
$Y\in\bound(\sH_A)$. Therefore, if $\mN$ is a channel, its
adjoint $\mN^*$ is a completely positive, unital (unit-preserving), i.e. $\mN^*(\openone_B)=\openone_A$, linear map (and vice versa).

Given a channel $(\sH_A,\sH_B,\mN)$, its
\emph{Stinespring isometric dilation}~\cite{stinespring} is given by a complementary (ancillary) quantum
system $\sH_E$ (the `environment') together with an isometry
$V:\sH_A\to\sH_B\otimes\sH_E$, $V^\dag V=\openone_A$, such that
\begin{equation*}
\mN(X)=\Tr_E[V X V^\dag],\quad\forall X\in\bound(\sH_A).
\end{equation*}
The Stinespring isometric dilation, which always exists, can be
considered to be essentially unique, in the sense that it is
unique up to isometric equivalences on $\sH_E$. This leads us to
define an essentially unique complementary channel
$(\sH_A,\sH_E,\mN_{\operatorname{env}})$ as follows~\cite{deg1,compl1}:
\begin{equation*}
\mN_{\operatorname{env}}(Y):=\Tr_B[V Y V^\dag].
\end{equation*}

\begin{definition}[Antidegradable channels]
	Given a channel $(\sH_A,\sH_B,\mN)$, let $(\sH_A,\sH_E,\mN_{\operatorname{env}})$ be its complementary channel. 
	$\mN$ is called
	\emph{antidegradable} if and only if there exists another channel
	$(\sH_E,\sH_B,\mD)$ such that
	\begin{equation*}
	\mN=\mD\circ\mN_{\operatorname{env}}.
	\end{equation*}
	(It is easy to verify that the property of being antidegradable does
	not depend on the particular Stinespring isometric dilation chosen to
	construct the complementary channel.)
\label{def:degradable}\end{definition}
In other words, an eavesdropper, Eve, who has access to the environment of an antidegradable channel, can perfectly simulate the output of the channel by means of a fixed post-processing which is independent of the input. In this sense, Eve always receives
\emph{more information} than the receiver Bob. As a straightforward consequence, antidegradable
channels turn out to have zero capacity for any
information-theoretic protocol that aims to put Bob in a position of advantage over eavesdropper.

Another notion we need is the following~\cite{SSW2008}:
\begin{definition}[$d$-dimensional symmetric channels]\label{def:symmetric}
For a given finite integer $d\ge 2$, let $\sH_A$ be a $\frac{d(d+1)}{2}$-dimensional Hilbert space, $\sH_B$ and $\sH_E$ be $d$-dimensional Hilbert spaces; moreover, let $V:\sH_A\to\sH_B\otimes\sH_E$ be any isometry embedding $\sH_A$ into the symmetric subspace $(\sH_B\otimes\sH_E)_{\operatorname{sym}}$; the channel $(\sH_A,\sH_B,\mS)$, defined by its action $\mS(\rho):=\Tr_E[V\rho V^\dag]$ for all $\rho\in\sS(\sH_A)$, is called a $d$-dimensional symmetric channel.	
\end{definition}
 It should be clear then that $d$-dimensional symmetric channels are, in particular, antidegradable, with the post-processing channel $\mD$ given by the identity map. The class of symmetric quantum channels has been identified as the quantum analogue of a \textit{public channel}~\cite{SSW2008,BO2012,BO2013}, 
since, for a symmetric channel, the receiver and eavesdropper receive 
the same output~\cite{footnote2}.

\subsection{Statement of the main result}
Before stating our main result we introduce some further definitions.

\begin{definition}
	A \emph{(finite) ensemble of quantum states} $\model{m}$ is defined as a triple $(\sH,\set{X},\ens{E})$, where $\sH$ is a finite-dimensional input Hilbert space,  $\set{X}=\{x\}$ is a finite indexing alphabet, and $\ens{E}=\{p_x,\rho^x\}_{x\in\set{X}}$ is a collection of quantum states $\rho^x\in\sS(\sH)$ and probabilities $p_x$. 
\end{definition}
Consider now a quantum channel $(\sH_A,\sH_B,\mN)$ and an ensemble $\model{m}=(\sH_A,\set{X},\ens{E})$. We can then imagine the situation in which a sender (say, Alice) chooses a letter $x\in\set{X}$ at random according to the probability distribution $p_x$, prepares a quantum system in the corresponding state $\rho_A^x$, and sends this through $\mN$ to a receiver (say, Bob), who has to guess the input letter chosen by Alice. This setup can be formally described as follows:
\begin{definition}[Dynamical guessing games]\label{def:dynamic-guessing}
	Let $(\sH_A,\sH_B,\mN)$ be a quantum channel,
	$(\sH_A,\set{X},\ens{E})$ an ensemble. The corresponding \textit{guessing game} is defined as the task of correctly guessing letter $x$ upon receiving $\mN(\rho^x_A)$. The optimal probability of winning the game is given by
	\begin{equation}\label{eq:guess-prob}
	p^*(\mN,\model{m}):=\max_{\povm{P}_B}\sum_{x\in\set{X}}p_x\Tr[P^x_B\ \mN(\rho^x_A)].
	\end{equation}
\end{definition}

Equation~(\ref{eq:guess-prob}) above measures `how good' is a given channel $\mN$ for communicating the information about $\set{X}$ encoded in $\model{m}$. Accordingly, given another channel $(\sH_A,\sH_{B'},\mM)$, with same input space but generally different output space, one can say that `$\mN$ is not worse than $\mM$ with respect to $\model{m}$' if $p^*(\mN,\model{m})\ge p^*(\mM,\model{m})$. By extending this definition to every possible finite ensemble, we obtain the following partial ordering relation between quantum channels:
\begin{definition}\label{def:more-info}
	Given two quantum channels with the same input space $(\sH_A,\sH_B,\mN_\alpha)$ and $(\sH_A,\sH_{B'},\mN_\beta)$, we say that `$\mN_\alpha$ is \textit{more informative} than $\mN_\beta$,' and denote it as $\mN_\alpha\supseteq\mN_\beta$, whenever $p^*(\mN_\alpha,\model{m})\ge p^*(\mN_\beta,\model{m})$, for all finite ensembles $\model{m}$ on $\sH_A$.
\end{definition}

Clearly, guessing games can be also played with more than one channel arranged `in parallel,' as follows. Consider for example two quantum channels $(\sH_A,\sH_B,\mN)$ and $(\sH_{A_0},\sH_{B_0},\mM)$ and an ensemble defined on the tensor product space $\sH_A\otimes\sH_{A_0}$, i.e. $\model{n}=(\sH_A\otimes\sH_{A_0},\set{X},\ens{E})$. Then, in analogy with~(\ref{eq:guess-prob}), we have
\begin{equation}
p^*(\mN\otimes\mM,\model{n})=\max_{\povm{P}_{BB_0}}\sum_{x\in\set{X}}p_x\Tr[P^x_{BB_0}\ (\mN\otimes\mM)(\rho^x_{AA_0})].
\end{equation}
It is important to stress that, as the input states $\rho^x_{AA_0}$ can be entangled, so the elements $P^x_{BB_0}$ of the decoding POVM are allowed to act globally on the output. By means of parallelized guessing games, a stronger partial ordering relation can be introduced as follows:
\begin{definition}[Strong information ordering]\label{def:strongly-more-info}
	Given two quantum channels with the same input space  $(\sH_A,\sH_B,\mN_\alpha)$ and $(\sH_A,\sH_{B'},\mN_\beta)$, we say that `$\mN_\alpha$ is \textit{strongly more informative} than $\mN_\beta$,' and denote it as \begin{equation*}
	\mN_\alpha\supseteq_{\operatorname{s}}\mN_\beta,
	\end{equation*}whenever $\mN_\alpha\otimes\mM\supseteq\mN_\beta\otimes\mM$, for all quantum side channels $(\sH_{A_0},\sH_{B_0},\mM)$.
\end{definition}
In the above definition, we allow the comparison between $\mN_\alpha$ and $\mN_\beta$ to be made in parallel with any possible quantum side channel $(\sH_{A_0},\sH_{B_0},\mM)$ considered as an auxiliary communication resource. It is often interesting, however, to constrain the side channel to belong to some restricted class of channels, typically with reduced communication capability. As a trivial example, Definition~\ref{def:more-info} can be considered as a special case of Definition~\ref{def:strongly-more-info}, in which side channels are restricted to those which map all input states to the same output state. Here, for reasons that will be clarified later, we are in particular interested in the case in which the quantum side channel is a symmetric channel $(\sH_{A_0},\sH_{B_0},\mS)$, as introduced in Definition~\ref{def:degradable}:
\begin{definition}[Weak information ordering]\label{def:weak-more-info}
	Given two quantum channels with the same input space  $(\sH_A,\sH_B,\mN_\alpha)$ and $(\sH_A,\sH_{B'},\mN_\beta)$, we write
	\[\mN_\alpha\supseteq_{\operatorname{w}}\mN_\beta,\]
	 whenever there exists a symmetric quantum side channels $(\sH_{A_0},\sH_{B_0},\mS)$, with $\sH_{B_0}\cong\sH_{B'}$, such that $\mN_\alpha\otimes\mS\supseteq\mN_\beta\otimes\mS$.
\end{definition}
Notice that the above definition relaxes Definition~\ref{def:strongly-more-info}, not only in that the comparison can be made just with respect to symmetric side channels (rather than \textit{any} side channel), but just with respect to \textit{some} symmetric side-channel (under the sole condition $\sH_{B_0}\cong\sH_{B'}$).

The main technical result of this paper is summarised in the following theorem, for which a proof will be given in Section~\ref{sec:back-to-channels}:
\begin{theorem}\label{theo:main}
	Let $(\sH_A,\sH_B,\mN_\alpha)$ and $(\sH_A,\sH_{B'},\mN_\beta)$ be two quantum channels with the same input space $\sH_A$. Then, the following are equivalent:
	\begin{enumerate}
		\item there exists a third quantum channel $(\sH_{B},\sH_{B'},\mD)$ such that $\mN_\beta=\mD\circ\mN_\alpha$;
		\item $\mN_\alpha\supseteqs\mN_\beta$;
		\item $\mN_\alpha\supseteqw\mN_\beta$.
			\end{enumerate}
\end{theorem}


An interesting interpretation of Theorem~\ref{theo:main} is obtained when $\mN_\beta$ and $\mN_\alpha$ are taken to be the channel $\mN$ (which we wish to characterize) and its corresponding complementary channel $\mNenv$, respectively. In this situation, consider the game-theoretic scenario~\ref{game} described in the Introduction, 
in which, at each turn of the game (corresponding to each use of the channel), Bob and Eve are asked to guess the input chosen by Alice. 
In this case, it is natural to require the side-channel ${\cal{S}}$ to be symmetric, so that it serves as a public channel~\cite{SSW2008,BO2012,BO2013}, 
since it conveys the same information to Bob and Eve.


Theorem~\ref{theo:main} then implies the following corollary which provides a complete characterization of antidegradable channels in the game-theoretic scenario~\ref{game}:

\begin{corollary}\label{coro:main}
	A channel is not antidegradable if and only if there exists an encoding strategy for Alice
which results in Bob winning the game \ref{game}.
\end{corollary}

The above corollary guarantees that any channel $\cal N$, as long as it is not antidegradable, puts Bob in a position of advantage with respect to Eve in the game~\ref{game}.

\section{From quantum channels to bipartite states...}\label{sec:states}

In this section we derive results pertaining to quantum states, which are analogues of the results stated in Theorem~\ref{theo:main} for quantum channels. We begin by recalling a fundamental relation, due to Choi~\cite{choi}, between bipartite states and channels.

\begin{theorem}[Choi Isomorphism]\label{theo:choi}
Fix an orthonormal basis $\{|i\>\}_{i=1}^d$ in  a finite-dimensional Hilbert space $\sH_A$ ($\dim\sH_A=d$). Define the standard maximally entangled state 
$|\Phi^+\>:=d^{-1/2}\sum_{i=1}^{d}|i\>\otimes|i\>\in\sH_A\otimes\sH_A$. Then, any channel $(\sH_A,\sH_B,\mN)$ defines a bipartite state $\rho_{AB}^\mN\in\sS(\sH_A\otimes\sH_B)$ with $\Tr_B[\rho_{AB}^\mN]=d^{-1}\openone_A$ via the relation:
	\[
	\rho_{AB}^\mN:=(\id\otimes\mN)(|\Phi^+\>\<\Phi^+|).
	\]
	Conversely, any bipartite state $\rho_{AB}\in\sS(\sH_A\otimes\sH_B)$ with $\Tr_B[\rho_{AB}]=d^{-1}\openone_A$ defines a channel $(\sH_A,\sH_B,\mN^\rho)$ via the relation:
	\[
	\mN^\rho(X):=d\Tr_A[(X^T\otimes\openone_B)\ \rho_{AB}],
	\]
	for all $X\in\bound(\sH_A)$,
	where the transposition is taken with respect to the fixed basis $\{|i\>\}_{i=1}^d$. The correspondence is one-to-one, i.e.
	\[
	d\Tr_A[(X^T\otimes\openone_B)\ \rho^\mN_{AB}]=\mN(X),
	\]
	and
	\[
	(\id\otimes\mN^\rho)(|\Phi^+\>\<\Phi^+|)=\rho_{AB}.
	\]
\end{theorem}

With the Choi isomorphism at hand, we will reformulate Theorem~\ref{theo:main} as a result about the comparison of quantum bipartite states, rather than channels, in the spirit of Ref.~\cite{q-black}.

\subsection{Statistical comparison of bipartite quantum states}

As in Ref.~\cite{shmaya,chefles,q-black}, we can characterize bipartite quantum states in terms of the following game-theoretical scenarios:

\begin{itemize}
	\item Quantum Statistical Decision Games: these are defined by
an outcome set $\set{X}=\{x\}$ and a family of self-adjoint
	operators $\{O^x_A\}_{x\in\set{X}}$; given a bipartite quantum state $\rho_{AB}\in\sS(\sH_A\otimes\sH_B)$, its payoff with respect to a quantum statistical decision game is given by
	\begin{equation}
	\max_{\povm{Q}_B}\sum_{x} \Tr[(O^x_A\otimes
	Q^x_{B})\ \rho_{AB}].
	\end{equation}
	
	\item Quantum Statistical Decision Problems: these are defined by two
	outcome sets $\set{X}=\{x\}$ and $\set{Y}=\{y\}$, a POVM
	$\{P^x_A\}_{x\in\set{X}}$, and a utility function
	$u:\set{X}\times\set{Y}\to\mathbb{R}$; given a bipartite quantum state $\rho_{AB}\in\sS(\sH_A\otimes\sH_B)$, its payoff with respect to a quantum statistical decision problem is given by
	\begin{equation}
	\max_{\povm{Q}_B}\sum_{x,y}u(x,y) \Tr[(P^x_A\otimes
	Q^y_{B})\ \rho_{AB}].
	\end{equation}
		
	\item Static Guessing Games: these are defined by an outcome set $\set{X}=\{x\}$ and a
	POVM $\{P^x_A\}_{x\in\set{X}}$; given a bipartite quantum state $\rho_{AB}\in\sS(\sH_A\otimes\sH_B)$, its payoff with respect to a guessing game is given by
	\begin{equation}\label{eq:static-guessing}
	\max_{\povm{Q}_B}\sum_{x} \Tr[(P^x_A\otimes
	Q^x_{B})\ \rho_{AB}].
	\end{equation}
	\end{itemize} 
	
	As done in Definition~\ref{def:more-info}, where quantum channels are compared with respect to their `utility' in playing guessing games, we can compare bipartite states in terms of their `utilities' in playing the three kinds of statistical games we introduced above. The following theorem states that, in the case in which we are to compare two bipartite states  $\rho_{AB}\in\sS(\sH_A\otimes\sH_B)$ and
	$\sigma_{AB'}\in\sS(\sH_A\otimes\sH_{B'})$, such that $\Tr_B\rho_{AB}=\Tr_{B'}\sigma_{AB'}$, the corresponding partial ordering relations are all equivalent.

	\begin{theorem}\label{theo:comparison}
		Let $\rho_{AB}\in\sS(\sH_A\otimes\sH_B)$ and
		$\sigma_{AB'}\in\sS(\sH_A\otimes\sH_{B'})$ be such that
		$\Tr_B\rho_{AB}=\Tr_{B'}\sigma_{AB'}$. Then, the following
		statements are equivalent:
		
		\begin{enumerate}
			
			\item Comparison by quantum statistical decision games. For
			any outcome set $\set{X}=\{x\}$ and for any set of self-adjoint
			operators $\{O^x_A\}_{x\in\set{X}}$,
			\begin{equation}\label{eq1}
			\max_{\povm{R}_B}\sum_{x} \Tr[(O^x_A\otimes
			R^x_{B})\ \rho_{AB}]\ge  \max_{\povm{Q}_{B'}}\sum_{x} \Tr[(O^x_A\otimes
			Q^x_{B'})\ \sigma_{AB'}];
			\end{equation}
			
			\item Comparison by quantum statistical decision problems. For any
			outcome sets $\set{X}=\{x\}$ and $\set{Y}=\{y\}$, for any POVM
			$\{P^x_A\}_{x\in\set{X}}$, and for any utility function
			$u:\set{X}\times\set{Y}\to\mathbb{R}$,
			\begin{equation}\label{eq2}
			\max_{\povm{R}_B}\sum_{x,y}u(x,y) \Tr[(P^x_A\otimes
			R^y_{B})\ \rho_{AB}]\ge  \max_{\povm{Q}_{B'}}\sum_{x,y}u(x,y) \Tr[(P^x_A\otimes
			Q^y_{B'})\ \sigma_{AB'}];
			\end{equation}
			
			\item Comparison by the Hahn-Banach separation theorem. For any outcome sets $\set{X}=\{x\}$ and $\set{Y}=\{y\}$, for any
			POVMs $\{P^x_A\}_{x\in\set{X}}$ and $\{Q^y_{B'}\}_{y\in\set{Y}}$, there
			exists a POVM $\{R^y_B\}_{y\in\set{Y}}$ such that
			\begin{equation}\label{eq3}
			\Tr[(P^x_A\otimes
			R^y_{B})\ \rho_{AB}]= \Tr[(P^x_A\otimes Q^y_{B'})\
			\sigma_{AB'}]\qquad\forall x,y;
			\end{equation}
			
			\item Comparison by guessing games. For any outcome set $\set{X}=\{x\}$ and for any
			POVM $\{P^x_A\}_{x\in\set{X}}$,
			\begin{equation}\label{eq4}
			\max_{\povm{R}_B}\sum_{x} \Tr[(P^x_A\otimes
			R^x_{B})\ \rho_{AB}]\ge  \max_{\povm{Q}_{B'}}\sum_{x} \Tr[(P^x_A\otimes
			Q^x_{B'})\ \sigma_{AB'}].
			\end{equation}
			
		\end{enumerate}
		
	\end{theorem}
	
	\begin{proof}
		The relation (1)$\Rightarrow$(2) holds because any specification of
		an outcome set $\set{X}$ together with a utility function $u$ defines,
		in particular, a set of self-adjoint operators
		$\{O^y_A\}_{y\in\set{Y}}$, by the summation
		$O^y_A:=\sum_xu(x,y)P^x_A$. Hence,
		\begin{eqnarray}
		{\hbox{RHS of \reff{eq2}}} &=&  \max_{\povm{Q}_{B'}}\sum_{y}
		\Tr[((\sum_x u(x,y)P^x_A)\otimes
		Q^y_{B'})\ \sigma_{AB'}]\nonumber\\
		&=& \max_{\povm{Q}_{B'}}\sum_{y}
		\Tr\left[(O^y_A \otimes
		Q^y_{B'})\ \sigma_{AB'}\right]\nonumber\\
		&\le &   \max_{\povm{R}_B}\sum_{y} \Tr\left[(O^y_A\otimes
		R^y_{B})\ \rho_{AB}\right]\nonumber\\
		&=&   \max_{\povm{R}_B}\sum_{x,y}u(x,y) \Tr[(P^x_A\otimes
		R^y_{B})\ \rho_{AB}]\nonumber\\
		&=& {\hbox{LHS of \reff{eq2}}},
		\end{eqnarray}
		where the inequality follows from (1) and the third equality follows from
		the definition of $O^y_{A}$.
		
		The relation (2)$\Leftrightarrow$(3) holds as a consequence of the
		separation theorem for convex sets (for a detailed discussion on this
		point see, for example, Ref.~\cite{q-black}).
		
		The relation (2)$\Rightarrow$(4) holds simply by taking
		$u(x,y)=\delta_{xy}$ in \reff{eq2}.
		
		The relation (4)$\Rightarrow$(1) (which would complete the proof of
		equivalence) can be established as follows: Given an outcome set
		$\set{X}$ and a set of self-adjoint operators $\{O^x_A\}_{x\in\set{X}}$,
		let us define the following operators for $x \in \set{X}$:
		\begin{equation*}
		P^x_A:=\frac 1{\lambda}\frac
		1{|\set{X}|}\left\{O^x_A+\lambda\openone_A-\frac 1{|\set{X}|}\Sigma_A\right\},
		\end{equation*}
		where $\Sigma_A:=\sum_xO^x_A$ and $0< \lambda < \infty$ is chosen such
		that $P^x_A \ge 0$ for all $x$. By construction $\sum_xP^x_A=\openone_A$, and hence
		$\{P^x_A\}_{x \in \set{X}}$ is a POVM. For each $x \in \set{X}$, then,
		\begin{equation}\label{ox}
		O^x_A = \lambda |\set{X}| P^x_A - \lambda \openone_A + \frac{1}{|\set{X}|}\Sigma_A.
		\end{equation}
		Substituting \reff{ox} on the RHS of \reff{eq1} we get
		\begin{eqnarray} {\hbox{RHS of \reff{eq1}}} &=& \max_{\povm{Q}_{B'}}\sum_{x}\Tr\left\{\left[\left(\lambda |\set{X}| P^x_A - \lambda \openone_A + \frac{1}{|\set{X}|}\Sigma_A\right) \otimes   Q^x_{B'}\right]\sigma_{AB'}\right\}
		\nonumber\\
		&=& \lambda |\set{X}| \max_{\povm{Q}_{B'}}\left\{\sum_{x}
		\Tr[(P^x_A\otimes
		Q^x_{B'})\ \sigma_{AB'}]\right\} - \lambda +  \frac{1}{|\set{X}|} \Tr \Sigma_A \rho_A,\nonumber\\
		&\le & \lambda |\set{X}| \max_{\povm{R}_B}\left\{\sum_{x}
		\Tr[(P^x_A\otimes
		R^x_{B})\ \rho_{AB}]\right\} - \lambda + \frac{1}{|\set{X}|} \Tr \Sigma_A \rho_A\nonumber\\
		&=& \max_{\povm{R}_B}\sum_{x}
		\Tr\left\{\left[\left(\lambda |\set{X}| P^x_A - \lambda \openone_A + \frac{1}{|\set{X}|}\Sigma_A\right) \otimes   R^x_{B}\right]\rho_{AB}\right\}\nonumber\\
		&=& \max_{\povm{R}_B}\sum_{x} \Tr[(O^x_A\otimes
		R^x_{B})\ \rho_{AB}]\label{second-last}\\
		&=& {\hbox{LHS of \reff{eq1}}}\nonumber,
		\end{eqnarray}
		where the second equality follows from the facts that $\sum_x Q^x_{B'}
		= \openone_{B'}$ and $\Tr_{B'} \sigma_{AB'} =\Tr_B\rho_{AB}\equiv \rho_A$, the inequality
		follows from \reff{eq4}, and~\reff{second-last} follows from \reff{ox}.
	\end{proof}

	\begin{remark}
		The condition $\Tr_B\rho_{AB}=\Tr_{B'}\sigma_{AB'}$ is crucial for
		the validity of Theorem~\ref{theo:comparison}. If this condition is
		dropped, statements (1), (2), and (3) are still equivalent, while
		statement (4) becomes only a \emph{necessary condition} for the
		validity of the previous three~\cite{q-black}.
	\end{remark}
	
	We can then introduce the following definition:
	\begin{definition}
		Let $\rho_{AB}\in\sS(\sH_A\otimes\sH_B)$ and
		$\sigma_{AB'}\in\sS(\sH_A\otimes\sH_{B'})$ be such that
		$\Tr_B\rho_{AB}=\Tr_{B'}\sigma_{AB'}$. We say that $\rho_{AB}$ is
		\emph{more informative} than $\sigma_{AB'}$, written as
		\begin{equation*}
		\rho_{AB}\supseteq_A\sigma_{AB'},
		\end{equation*}
		if and only if any one of the four statements in
		Theorem~\ref{theo:comparison} holds.
	\end{definition}
	
Finally, as channels can be arranged in parallel and used to play parallelized guessing games (see Definitions~\ref{def:strongly-more-info} and~\ref{def:weak-more-info}), bipartite states too can be put in parallel and compared in a similar manner. For example, given two quantum states $\rho_{AB}\in\sS(\sH_A\otimes\sH_B)$ and
$\sigma_{AB'}\in\sS(\sH_A\otimes\sH_{B'})$ such that
$\Tr_B\rho_{AB}=\Tr_{B'}\sigma_{AB'}$, let $\omega_{A_0B_0}\in\sS(\sH_{A_0}\otimes\sH_{B_0})$ be a third auxiliary bipartite state. Then we can write
\[
\rho_{AB}\otimes\omega_{A_0B_0}\supseteq_{AA_0}\sigma_{AB'}\otimes\omega_{A_0B_0},
\] 
with the meaning that for any outcome set $\set{X}=\{x\}$ and for any POVM $\{P^x_{AA_0}\}_x$,
\[
\max_{\povm{R}_{BB_0}}\sum_{x} \Tr[(P^x_{AA_0}\otimes
R^x_{BB_0})\ (\rho_{AB}\otimes\omega_{A_0B_0})]\ge  \max_{\povm{Q}_{B'B_0}}\sum_{x} \Tr[(P^x_{AA_0}\otimes
Q^x_{B'B_0})\ (\sigma_{AB'}\otimes\omega_{A_0B_0})].
\]
The above equation directly generalizes Eq.~(\ref{eq4}) in Theorem~\ref{theo:comparison}. Along the same line, Eqs.~(\ref{eq1}), (\ref{eq2}), and~(\ref{eq3}) can also be generalized.
	
	\subsection{Local degradability of bipartite states}
	
	Another partial ordering relation between bipartite states can be introduced as follows:
	
	\begin{definition}[Local degradability]\label{def:suff}
		Given two quantum states $\rho_{AB}$ and $\sigma_{AB'}$ such that
		$\Tr_B\rho_{AB}=\Tr_{B'}\sigma_{AB'}$, we say that $\rho_{AB}$
		can be \emph{locally degraded} to $\sigma_{AB'}$, written as
		\begin{equation}
		\rho_{AB}\suff\sigma_{AB'},
		\end{equation}
		if and only if there exists a channel $(\sH_B,\sH_{B'},\mD)$ such that
		\begin{equation}
		\sigma_{AB'}=(\id_A\otimes\mD_B)(\rho_{AB}).
		\end{equation}
	\end{definition}
	
	In Ref.~\cite{q-black}, a fundamental equivalence relation between the two orderings $\supseteq$ and $\suff$ is proved. In what follows, we introduce all the ideas we need in order to adapt the equivalence relation of~\cite{q-black} to the present case.
	
	\begin{definition}[Local state space and complete states~\cite{faithful,q-black}]\label{def:complete} Given a
		bipartite state $\rho_{AB}\in\sS(\sH_A\otimes\sH_B)$, its
		\emph{local state space} $\sS_B(\rho_{AB})\subseteq\sS(\sH_B)$ is the
		convex set defined as follows:
		\begin{equation*}
		\sS_B(\rho_{AB})=\sS(\sH_B)\cap\left\{\Tr_A[(P_A\otimes\openone_B)\
		\rho_{AB}]\mid 0\le P_A\in\bound(\sH_A)\right\}.
		\end{equation*}
		Whenever $\sS_B(\rho_{AB})$ contains $(\dim\sH_B)^2$
		linearly independent elements, then $\rho_{AB}$ is said to be $B$-\emph{complete} (or, simply, \emph{complete}).
	\end{definition}
	Examples of complete bipartite states in $\sS(\sH_A\otimes\sH_B)$ are given by states of the form $p|\Phi^+_{AB}\>\<\Phi^+_{AB}|+\frac{(1-p)}{d_Ad_B}\openone_{AB}$, where $|\Phi^+_{AB}\>$ is a maximally entangled state in $\sH_A\otimes\sH_B$, for any $0<p\le 1$. We now prove a fact that will turn out to be useful later on:
\begin{lemma}\label{lemma:complete}
	A bipartite state $\rho_{AB}\in\sS(\sH_A\otimes\sH_B)$ is
	$B$-complete if and only if there exists a POVM $\{P^x_A\}_x$ on
	$\sH_A$ such that the set $\{\rho^x_B\}_x$, where
	$\rho^x_B:=\Tr_A[(P^x_A\otimes\openone_B)\ \rho_{AB}]$, contains $(\dim\sH_B)^2$
	linearly independent elements.
\end{lemma}

\begin{proof}
	Suppose that there exists a POVM $\{P^x_A\}_x$ on $\sH_A$ such that
	the set $\{\rho^x_B\}_x$, where
	$\rho^x_B:=\Tr_A[(P^x_A\otimes\openone_B)\ \rho_{AB}]$, contains $(\dim\sH_B)^2$
	linearly independent elements. Then, define
	positive operators as follows:
	\begin{equation*}
	\tilde P^x_A:=\frac{P^x_A}{\Tr[\rho^x_B]},
	\end{equation*}
	and, correspondingly, $\tilde\rho^x_B:=\Tr_A[(\tilde
	P^x_A\otimes\openone_B)\ \rho_{AB}]$. Clearly, all
	$\tilde\rho^x_B$ belong to $\sS_B(\rho_{AB})$, and they are linearly independent if and only if the $\rho^x_B$ are. Therefore $\sS_B(\rho_{AB})$ contains $(\dim\sH_B)^2$
	linearly independent elements, i.e.
	$\rho_{AB}$ is $B$-complete.
	
	Conversely, suppose that $\rho_{AB}$ is $B$-complete. Then, there
	exist $(\dim\sH_B)^2$ positive operators $P^x_A$ such that all
	$\rho^x_B=\Tr_A[(P^x_A\otimes\openone_B)\
	\rho_{AB}]\in\sS_B(\rho_{AB})$ are linearly independent. However,
	$\sum_xP^x_A\neq\openone_A$, i.e. the operators $\{P^x_A\}_x$, even
	though positive, do not constitute, in general, a POVM. Let then $\lambda$
	be any strictly positive number such that $\lambda\sum_xP^x_A\le\openone_A$, and
	define $\tilde P^x_A:=\lambda P^x_A$, and $\tilde
	P^\infty_A:=\openone_A-\sum_x \tilde P^x_A$. Then, the set $\{\tilde
	P^x_A\}_x\cup\{\tilde P^\infty_A\}$ constitutes a well defined
	POVM. Define also $\tilde\rho^x_B:=\Tr_A[(\tilde
	P^x_A\otimes\openone_B)\ \rho_{AB}]$. Since the first $(\dim\sH_B)^2$ elements of $\{\tilde\rho^x_B\}$ are linearly independent if and only if the
	$\rho^x_B$ are, we have that the set $\{\tilde\rho^x_B\}$ surely contains $(\dim\sH_B)^2$ linearly independent elements.
\end{proof}
	
	\begin{theorem}[Comparison of bipartite quantum
		states~\cite{q-black,semi-quantum}]\label{theo:equiv}
		Given two bipartite quantum states
		$\rho_{AB}\in\sS(\sH_A\otimes\sH_B)$ and
		$\sigma_{AB'}\in\sS(\sH_A\otimes\sH_{B'})$ with
		$\Tr_B\rho_{AB}=\Tr_{B'}\sigma_{AB'}$, the following are equivalent:
		\begin{enumerate}
			\item $\rho_{AB}\suff\sigma_{AB'}$;
			\item for any $\sH_{A_0}$, any $\sH_{B_0}$, and any auxiliary bipartite state
			$\omega_{A_0B_0}\in\sS(\sH_{A_0}\otimes\sH_{B_0})$,\newline
			\mbox{$\rho_{AB}\otimes\omega_{A_0B_0}\supseteq_{AA_0}
				\sigma_{AB'}\otimes\omega_{A_0B_0}$;}
			\item for some $B_0$-complete state $\omega_{A_0B_0}$, with
			$\sH_{B_0}\cong\sH_{B'}$,
			\newline\mbox{$\rho_{AB}\otimes\omega_{A_0B_0}\ \supseteq_{AA_0}\
				\sigma_{AB'}\otimes\omega_{A_0B_0}$.}
		\end{enumerate}
	\end{theorem}
	
	\begin{proof}
		We begin by noticing that the implication (1)$\Rightarrow$(2) is a
		trivial consequence of the fact that, if $\sigma_{AB'}=(\id_A\otimes\mD_B)(\rho_{AB})$ for some channel $\mD_B:\bound(\sH_B)\to\bound(\sH_{B'})$, the action of any POVM
		$\{Q^x_{B'B_0}\}_x$ on $\sigma_{AB'}\otimes\omega_{A_0B_0}$ can be
		exactly simulated on $\rho_{AB}\otimes\omega_{A_0B_0}$ by using the
		POVM $\{(\mD_{B'}^*\otimes\id_{B_0})(Q^x_{B'B_0})\}_x$, where we denoted
		by $\mD_{B'}^*:\bound(\sH_{B'})\to\bound(\sH_B)$ the Heisenberg dual
		of $\mD_B$.
		
		Also the implication (2)$\Rightarrow$(3) is trivial.
		
		We are then left to prove that (3)$\Rightarrow$(1). In order to do
		so, we consider two auxiliary Hilbert spaces $\sH_{A_0}$ and
		$\sH_{B_0}$, such that $\sH_{B_0}\cong \sH_{B'}$, and
		a $B_0$-complete state $\omega_{A_0B_0}$ (see
		Def.~\ref{def:complete}). We then consider, in particular, the
		following measurement on the composite state
		$\sigma_{AB'}\otimes\omega_{A_0B_0}$:
		\begin{equation*}
		\Tr[(\Upsilon^y_A\otimes \Xi^x_{A_0}\otimes B^z_{B'B_0})\ (\sigma_{AB'}\otimes\omega_{A_0B_0})],
		\end{equation*}
		where
		\begin{itemize}
			\item $\{\Upsilon^y_A\}_y$ is an informationally complete POVM on
			$\sH_A$ (i.e. any operator in $\bound(\sH_A)$ can be written as a linear combination of its elements);
			\item $\{\Xi^x_{A_0}\}_x$ is the POVM on $\sH_{A_0}$, whose existence
			is guaranteed by Lemma~\ref{lemma:complete}, inducing a complete set
			of linearly independent reduced (subnormalised) states
			$\omega_{B_0}^x=\Tr_{A_0}[(\Xi^x_{A_0}\otimes\openone_{B_0})\
			\omega_{A_0B_0}]$ on $\sH_{B_0}$;
			\item $\{B^z_{B'B_0}\}_z$ is a generalised Bell measurement on
			$\sH_{B'}\otimes\sH_{B_0}\cong\sH_{B'}^{\otimes 2}$ (i.e. a complete set of $(\dim\sH_{B'})^2$ orthogonal maximally entangled states).
		\end{itemize}
		
		First of all, we know that, by Theorem~\ref{theo:comparison}
		condition~3, there exists a POVM $\{R^z_{BB_0}\}_z$ such that
		\begin{equation*}
		\Tr[(\Upsilon^y_A\otimes \Xi^x_{A_0}\otimes R^z_{BB_0})\ (\rho_{AB}\otimes\omega_{A_0B_0})]=  \Tr[(\Upsilon^y_A\otimes \Xi^x_{A_0}\otimes B^z_{B'B_0})\ (\sigma_{AB'}\otimes\omega_{A_0B_0})],
		\end{equation*}
		for every triple $(x,y,z)$. Then, by first performing the
		trace over $\sH_{A_0}$, we obtain the following identity:
		\begin{equation}\label{eq:passaggio}
		\Tr[(\Upsilon^y_A\otimes R^z_{BB_0})\ (\rho_{AB}\otimes\omega^x_{B_0})]=  \Tr[(\Upsilon^y_A\otimes B^z_{B'B_0})\ (\sigma_{AB'}\otimes\omega^x_{B_0})],
		\end{equation}
		where, as we noticed above,
		$\Span\{\omega^x_{B_0}\}=\bound(\sH_{B_0})$.
		
		We now introduce another Hilbert space $\sH_{B_1}\cong\sH_{B_0}\cong\sH_{B'}$ and fix orthonormal bases $\{|\alpha^i\>\}_i$ and
		$\{|\beta^j\>\}_j$ for $\sH_{B_1}$ and $\sH_{B_0}$, respectively. Further, let the standard maximally entangled state in
		$\sH_{B_1}\otimes\sH_{B_0}$ be given by
		\begin{equation*}
		|\Phi^+_{B_1B_0}\>:=d^{-1/2}\sum_i|\alpha^i_{B_1}\>\otimes|\beta^i_{B_0}\>,
		\end{equation*}
		where $d:=\dim\sH_{B_1}=\dim\sH_{B_0}=\dim\sH_{B'}$. Let us, moreover,
		define the operators
		\begin{equation*}
		\begin{split}
		\Omega^x_{B_1}&=d^2\Tr_{B_0}[(\openone_{B_1}\otimes\omega^x_{B_0})\
		|\Phi^+_{B_1B_0}\>\<\Phi^+_{B_1B_0}|]\\
		&=d\ \left(\omega^x_{B_1}\right)^T,
		\end{split}
		\end{equation*}
		where the transposition is made with respect to the basis chosen in
		the definition of $|\Phi^+_{B_1B_0}\>$. Clearly,
		$\Span\{\omega^x_{B_0}\}=\bound(\sH_{B_0})$ implies that
		$\Span\{\Omega^x_{B_1}\}=\bound(\sH_{B_1})$, since neither the
		transposition nor the multiplication by a non-zero scalar affect the
		property of being linearly independent. It is moreover easy to verify (even by direct inspection) that
		\[
		\Tr_{B_1}[(\Omega^x_{B_1}\otimes\openone_{B_0})\ |\Phi^+_{B_1B_0}\>\<\Phi^+_{B_1B_0}|]=\omega^x_{B_0}
		\]
		for all $x$.
		
		Going back to Eq.~(\ref{eq:passaggio}), we can therefore rewrite it as:
		\begin{equation*}
		\Tr[(\Upsilon^y_A\otimes\Omega^x_{B_1}\otimes R^z_{BB_0})\
		(\rho_{AB}\otimes|\Phi^+\>\<\Phi^+|_{B_1B_0})]=  \Tr[(\Upsilon^y_A\otimes\Omega^x_{B_1}\otimes
		B^z_{B'B_0})\ (\sigma_{AB'}\otimes|\Phi^+\>\<\Phi^+|_{B_1B_0})].
		\end{equation*}
		Since both $\{\Upsilon^y_A\}_y$ and $\{\Omega^x_{B_1}\}_x$ are a
		complete set of linearly independent operators~\cite{footnote3}, the above equality is,
		in fact, an operator identity:
		\begin{equation}\label{eq:tr-parz}
		\Tr_{BB_0}[(\openone_{AB_1}\otimes R^z_{BB_0})\
		(\rho_{AB}\otimes|\Phi^+\>\<\Phi^+|_{B_1B_0})]=  \Tr_{B'B_0}[(\openone_{AB_1}\otimes B^z_{B'B_0})\  (\sigma_{AB'}\otimes|\Phi^+\>\<\Phi^+|_{B_1B_0})],
		\end{equation}
		for all $z$.
		
		We recall now that the POVM $\{B^z_{B'B_0}\}_z$, appearing on the
		right-hand side of the above equation, has been chosen to constitute a
		generalised Bell measurement on
		$\sH_{B'}\otimes\sH_{B_0}\cong\sH_{B'}^{\otimes 2}$. Therefore, the
		protocol of quantum teleportation provides unitary operators
		$U^z:\sH_{B_1}\to\sH_{B'}$ such that
		\begin{equation*}
		\sum_z(\openone_A\otimes
		U^z_{B_1}) \Big\{\Tr_{B'B_0}\left[(\openone_A\otimes
		\openone_{B_1}\otimes B^z_{B'B_0})\  (\sigma_{AB'}\otimes|\Phi^+\>\<\Phi^+|_{B_1B_0})\right]\Big\} (\openone_A\otimes U^z_{B_1})^\dag=\sigma_{AB'}.
		\end{equation*}
		Then, by defining a CPTP map
		$\mD:\bound(\sH_B)\to\bound(\sH_{B'})$:
		\begin{equation*}
		\mD(X_B):=\sum_{z} U^z_{B_1} \Big\{\Tr_{BB_0}\left[(\openone_{B_1}\otimes R^z_{BB_0})\
		(X_B\otimes|\Phi^+\>\<\Phi^+|_{B_1B_0})\right]\Big\} (U^z_{B_1})^\dag,
		\end{equation*}
		for all $X_B\in\bound(\sH_B)$, we arrive at
		\begin{equation*}
		(\id_A\otimes\mD)(\rho_{AB})=\sigma_{AB'},
		\end{equation*}
		i.e. $\rho_{AB}\suff\sigma_{AB'}$.
	\end{proof}
	
	\section{...and back to channels}\label{sec:back-to-channels}
	
	The starting observation is that, due to the invertibility of the Choi isomorphism (Theorem~\ref{theo:choi}), a channel $\mN$ can be degraded to another channel $\mM$ (i.e. there exists a third channel $\mD$ such that $\mM=\mD\circ\mN$) if and only if the bipartite state $\rho^{\mN}$ can be locally degraded to $\rho^{\mM}$, in the sense of Definition~\ref{def:suff}. However, before being able to translate Theorem~\ref{theo:equiv} into its analogue for channels, we first have to understand what sort of channels induce complete (in the sense of Definition~\ref{def:complete}) Choi states. The answer is given by the following definition:

\begin{definition}[Complete channels]\label{def:complete-ch}
	A channel $(\sH_A,\sH_B,\mN)$ is said to be
	\emph{complete} whenever its range contains $(\dim\sH_B)^2$
	linearly independent elements.
\end{definition}

Other than the trivial example of the identity channel, another, more interesting class of channels that are complete is given by $d$-dimensional symmetric channels of Definition~\ref{def:symmetric}.

\begin{lemma}\label{lemma:complete-channels}
	A channel is complete if and only if its associated Choi state is complete, in the sense of Def.~\ref{def:complete}. In particular, all $d$-dimensional symmetric channels are complete together with their associated Choi states.
\end{lemma}

\begin{proof}
 Let $(\sH_A,\sH_B,\mN)$ be a complete channel. By definition there exist $(\dim\sH_B)^2$ input states $\rho^i_A$ such that the set $\{\mN(\rho^i_A):1\le i\le(\dim\sH_B)^2\}$ spans the whole $\bound(\sH_B)$.
 
 We now recall the fact, sometimes referred to as \textit{steering}~\cite{schrodinger}, that, for any state $\rho_A$, there exists an operator $P_{\tilde A}>0$ such that $\rho_A=\Tr_{\tilde A}[(P_{\tilde A}\otimes\openone_A)\ |\Phi^+_{\tilde AA}\>\<\Phi^+_{\tilde AA}|]$, where $\sH_{\tilde A}\cong\sH_A$ and $|\Phi^+_{\tilde AA}\>$ is a maximally entangled state in $\sH_{\tilde A}\otimes\sH_A$. Therefore, for any given channel $(\sH_A,\sH_B,\mN)$, its Choi state $\rho^\mN_{AB}$ is constructed so that, for any input state $\rho_A$, there exists an operator $P_A>0$ such that $\mN(\rho_A)=\Tr_{A}[(P_{A}\otimes\openone_B)\ \rho^\mN_{AB}]$. In turn, this implies that, whenever the channel $(\sH_A,\sH_B,\mN)$ is complete, there exists a set of operators $\{P^i_A>0:1\le i\le(\dim\sH_B)^2\}$ such that the set $\{\Tr_A[(P^i_A\otimes\openone_B)\ \rho^\mN_{AB}]:1\le i\le (\dim\sH_B)^2\}$ spans the whole $\bound(\sH_B)$, i.e. the bipartite state $\rho^\mN_{AB}$ is complete, according to Definition~\ref{def:complete}.

Conversely, suppose that the Choi state $\rho^\mN_{AB}$, associated with a channel $(\sH_A,\sH_B,\mN)$, is ($B$-)complete, in the sense of Definition~\ref{def:complete}. Then, by definition, there exist $(\dim\sH_B)^2$ operators $P_A^i>0$ such that the states defined as $\rho^i_B:=\Tr_A[(P^i_A\otimes\openone_B)\ \rho^\mN_{AB}]$ are all linearly independent in $\bound(\sH_B)$. On the other hand,  $\rho^i_B=\Tr_A[(P^i_A\otimes\openone_B)\ \rho^\mN_{AB}]=\Tr_{\tilde{A}}[(P^i_{\tilde{A}}\otimes\openone_B)\ (\id_{\tilde{A}}\otimes\mN_{A})(|\Phi^+_{\tilde{A}A}\>\<\Phi^+_{A\tilde{A}}|)]=\mN_A(\rho^i_A)$, where $\rho^i_A:=\Tr_{\tilde{A}}[(P^i_{\tilde{A}}\otimes\openone_B)\ |\Phi^+_{\tilde{A}A}\>\<\Phi^+_{\tilde{A}A}|]$, i.e., all $\rho^i_B$ belong to the range of the channel $(\sH_A,\sH_B,\mN)$, meaning that $\mN$ is complete in the sense of Definition~\ref{def:complete-ch}.
 
 Let us now turn to the special case of $d$-dimensional symmetric channels, as introduced in Definition~\ref{def:symmetric}. We just show that the channels are complete; the completeness of the corresponding Choi states then comes automatically. Consider therefore any $d$-dimensional symmetric channel $\mS$, defined by two $d$-dimensional Hilbert spaces $\sH_B$ and $\sH_E\cong\sH_B$, a $\frac{d(d+1)}{2}$-dimensional Hilbert space $\sH_A$, and an isometry $V:\sH_A\to(\sH_B\otimes\sH_E)_{\operatorname{sym}}$. Choose now $(\dim\sH_B)^2$ vectors $\{|\phi^i_B\>\}$ in $\sH_B$ such that the corresponding rank-one states $|\phi^i_B\>\<\phi^i_B|$ are all linearly independent in $\bound(\sH_B)$. Since $|\phi^i_B\>\otimes|\phi^i_E\>\in(\sH_B\otimes\sH_E)_{\operatorname{sym}}$ for all $i$, all the $(\dim\sH_B)^2$ pure states $\Tr_E[|\phi^i_B\>\<\phi^i_B|\otimes|\phi^i_E\>\<\phi^i_E|]=|\phi^i_B\>\<\phi^i_B|$ are possible outputs of $\mS$, i.e., $\mS$ is complete. Therefore, its associated Choi state $\omega_{AB}^{\mS}:=(\id\otimes\mS)(|\Phi^+\>\<\Phi^+|)=\Tr_E[(\openone\otimes V)(|\Phi^+\>\<\Phi^+|)(\openone\otimes V^\dag)]$ is a complete state.
\end{proof}

We are now ready to prove Theorem~\ref{theo:main}. In fact, we will do this indirectly, by proving that Theorem~\ref{theo:main} is nothing but Theorem~\ref{theo:equiv} formulated for a channel, rather than for a bipartite quantum state.

\begin{proof}[Proof of Theorem~\ref{theo:main}]
	Since implications (1) $\Rightarrow$ (2), and (2) $\Rightarrow$ (3) are trivial, we will focus only on the implication (3) $\Rightarrow$ (1).
	
	In order to prove the implication (3) $\Rightarrow$ (1), first of all we notice that, given two channels $(\sH_A,\sH_B,\mN_{\alpha})$ and $(\sH_A,\sH_{B'},\mN_{\beta})$, the Choi isomorphism (Theorem~\ref{theo:choi}) provides two bipartite states $\rho^\alpha_{AB}:=(\id\otimes\mN^{\alpha})(|\Phi^+\>\<\Phi^+|)$ and $\sigma^\beta_{AB'}:=(\id\otimes\mN^{\beta})(|\Phi^+\>\<\Phi^+|)$ such that $\Tr_B\rho^\alpha_{AB}=\Tr_{B'}\sigma^\beta_{AB'}=d_A^{-1}\openone_A$. We can therefore apply Theorem~\ref{theo:equiv} to $\rho^\alpha_{AB}$ and $\sigma^\beta_{AB'}$.
	
	Since point (3) of Theorem~\ref{theo:equiv} requires the comparison to be performed with some additional complete bipartite state, we can take  the state $\omega_{A_0B_0}$ appearing in point (3) of Theorem~\ref{theo:equiv} to be, in fact, the Choi state corresponding to a $d_{B_0}$-symmetric channel, which we know it is complete as a consequence of Lemma~\ref{lemma:complete-channels}.
	
	We then notice that, playing a `static' guessing game, as defined in Eq.~(\ref{eq:static-guessing}), with some Choi state is statistically equivalent to playing a `dynamic' guessing game with the corresponding channel, as described in Definition~\ref{def:dynamic-guessing}. The relation between the two approaches is again given by steering. As already noticed in the proof of Lemma~\ref{lemma:complete-channels}, for any given channel $(\sH_A,\sH_B,\mN)$, its Choi state $\rho^\mN_{AB}$ is constructed so that, for any ensemble $(\sH_A,\set{X},\{p_x,\rho^x_A\})$ there exists a POVM $\{P^x_A\}$ such that $p_x\mN(\rho^x_A)=\Tr_{A}[(P^x_{A}\otimes\openone_B)\ \rho^\mN_{AB}]$ for all $x$.
	
	It is therefore clear that point (3) in Theorem~\ref{theo:equiv} is completely equivalent (in fact, just a reformulation) of point (3) in Theorem~\ref{theo:main}. Since point (3) in Theorem~\ref{theo:equiv} is also equivalent to point (1)  in Theorem~\ref{theo:equiv}, we are left to show that point (1)  in Theorem~\ref{theo:equiv} is just a reformulation of point (1) in Theorem~\ref{theo:main}. The logical steps are summarized as follows:
	\begin{equation*}
	\textrm{Thm.~\ref{theo:main}, point (3)}\quad\Leftrightarrow\quad\textrm{Thm.~\ref{theo:equiv}, point (3)}\quad\Leftrightarrow\quad\textrm{Thm.~\ref{theo:equiv}, point (1)}\quad\Leftrightarrow\quad\textrm{Thm.~\ref{theo:main}, point (1)},
	\end{equation*}
	where the first equivalence has been proved above, the second equivalence is in the statement of Theorem~\ref{theo:equiv}, and only the last equivalence is left to be proved. But this is a simple consequence of the fact that Choi's correspondence is one-to-one, therefore two channels $\mN_\alpha$ and $\mN_\beta$ are such that there exists a third channel $\mD$ with $\mN_\beta=\mD\circ\mN_\alpha$, if and only if $\rho^{\mN_\alpha}_{AB}\succ\rho^{\mN_\beta}_{AB'}$.
	 \end{proof}

\section{Further implications of Theorem~\ref{theo:main}}\label{sec:further}

One can extend the results of Theorem~\ref{theo:main} to convex combinations of channels. It was shown in~\cite{SmithSmolin08} that degradable channels and degradable extensions have especially nice properties which prove to be useful when evaluating their quantum and private capacities. These properties are also reflected in the game-theoretic framework. In particular, the following two corollaries show how to compare and combine convex combinations of degradable channels and their extensions in this framework.

\begin{corollary}
	Consider a channel $(\sH_A,\sH_B,\mN_\beta)$, and a sequence of channels $(\sH_A,\sH_B,\mN_i)$, $i=1,\dots,n$, such that $\mN_\beta = \sum_i p_i \mN_i$. Assume that each $\mN_i$ is degradable, with the corresponding degrading map is given by $(\sH_B,\sH_B^{'},\mD_i)$. Define the flagged version of the convex combination of $\mN_i$ as ${\cal T} = \sum_i p_i \mN_i\otimes |i\rangle\langle i|$. Then, $\cal T$ is strongly more informative than $\mN_\beta$, i.e.
${\cal T}\supseteqs\mN_\beta$.

\end{corollary}
\begin{proof}
It was proven in~\cite{SmithSmolin08} that $\cal T$ is a degradable extension of $\mN_\beta$. Then the corollary follows after applying Theorem~\ref{theo:main}.
\end{proof}

\begin{corollary}
Consider two channels $(\sH_A,\sH_B,\mN_i)$, $i=1,2$, for each of which there exist $(\sH_A,\sH_B^{'},{\cal T}_i)$ and $(\sH_B^{'},\sH_B,\mD_i)$ such that $\mN_i = \mD_i\circ {\cal T}_i$ for $i=1,2$.
Then, for ${\cal T} = p {\cal T}_1\otimes |1\rangle\langle 1| + (1-p){\cal T}_2\otimes |2\rangle\langle 2|$ and $\mN = p\mN_1 + (1-p)\mN_2$ we have that $\cal T$ is strongly more informative than $\cal N$: ${\cal T}\supseteqs\mN$.

\end{corollary}
\begin{proof}
It is sufficient to observe that $\cal T$ is a degradable extension of $\mN$~\cite{SmithSmolin08}.  Then the corollary follows after applying Theorem~\ref{theo:main}.
\end{proof}

\section{Conclusions}\label{sec:conclude}
We introduced a game-theoretic framework~\ref{game} which allowed us to 
derive a necessary and sufficient condition for a channel to be antidegradable.
We showed that for any channel which is not antidegradable, there exists 
an encoding strategy for which such a channel provides a strict advantage 
for the two players over the adversary in the guessing game that we 
defined. The key ingredients in the proof of this result are the tools of statistical comparison of bipartite 
quantum states, and the Choi isomorphism.

The exact relationship between our game-theoretic framework and the standard information-theoretic framework remains to be explored. 
It would be interesting to
see whether any inference about the quantum or private capacity of a 
quantum channel could be made from results obtained in our game-theoretic framework.

Another direction worth pursuing is one which involves devising game-theoretic characterizations of other classes of quantum channels, since this might lead to a better 
understanding of the structure of zero-capacity channels. 
It would also be interesting to explore the connections between our game-theoretic approach and other incapacity tests~\cite{smith_detecting_2012} for quantum channels. 

\section*{Acknowledgements}
The authors are grateful to Michele Dall'Arno for suggesting an improvement 
to their previous proof of Lemma~\ref{lemma:complete}.
S.S. acknowledges the support of Sidney Sussex College.

\begin{thebibliography}{99}
	
	\bibitem{devetak} I. Devetak, \textit{The private classical capacity and quantum capacity of a quantum channel}. IEEE Transactions on Information Theory \textbf{51} ,  Issue 1, pp 44-55 (January 2005).
	
	\bibitem{footnote1} When the viceversa is true, i.e., when a post-processing of the channel's output can simulate the output to the environment, we speak of \textit{degradable channels}~\cite{deg1,deg2}.
	
	\bibitem{deg1} I. Devetak and P. W. Shor,
	\textit{The Capacity of a Quantum Channel for Simultaneous Transmission of Classical and Quantum Information}. 	Communications in Mathematical Physics
	 \textbf{256}, Issue 2, pp 287-303 (June 2005).
	
		\bibitem{deg2} T. S. Cubitt, M.-B. Ruskai, and G. Smith, \textit{The structure of degradable quantum channels}. J. Math. Phys. \textbf{49}, 102104 (2008).

\bibitem{wootters_single_1982}
{ W.~K. Wootters and W.~H. Zurek}, {\em A single quantum cannot be cloned},
  , Published online: 28 October 1982; {\textbar} doi:10.1038/299802a0, 299
  (1982), pp.~802--803.
  
  \bibitem{horodecki_binding_1999}
{ P.~Horodecki, M.~Horodecki, and R.~Horodecki}, {\em Binding entanglement
  channels}, quant-ph/9904092,  (1999).
\newblock {J.Mod.Opt.} 47 (2000) 347-354.

\bibitem{li_private_2009}
{ K.~Li, A.~Winter, X.~Zou, and G.~Guo}, {\em Private capacity of quantum
  channels is not additive}, Physical Review Letters, 103 (2009), p.~120501.

\bibitem{smith_extensive_2009}
{ G.~Smith and J.~A. Smolin}, {\em Extensive nonadditivity of privacy},
  Physical Review Letters, 103 (2009), p.~120503.
  
\bibitem{smith_detecting_2012}
{ G.~Smith and J.~A. Smolin}, {\em Detecting incapacity of a quantum
  channel}, Physical Review Letters, 108 (2012), p.~230507.


\bibitem{noise_cap}
{ F.~G. S.~L. Brand\~ao, J.~Oppenheim, and S.~Strelchuk}, {\em When does
  noise increase the quantum capacity?}, Phys. Rev. Lett., 108 (2012),
  p.~040501.
  
\bibitem{smith_quantum_2011}
{ G.~Smith, J.~A. Smolin, and J.~Yard}, {\em Quantum communication with
  gaussian channels of zero quantum capacity}, Nature Photonics, 5 (2011),
  pp.~624--627.

\bibitem{smith_quantum_2008}
{ G.~Smith and J.~Yard}, {\em Quantum communication with zero-capacity
  channels}, Science, 321 (2008), pp.~1812--1815.

\bibitem{stinespring}
{ W.~F. Stinespring}, {\em Positive functions on {$C^*$}-algebras}, Proc.
  Amer. Math. Soc., 6 (1955), pp.~211--216.

\bibitem{compl1}   A. S. Holevo, \textit{On complementary channels and the additivity problem}. Probab.
Theory and Appl. \textbf{51}, 133-143 (2005).

\bibitem{shmaya} E~Shmaya, \emph{Comparison of information structures
    and completely positive maps}. J. Phys. A: Math. and Gen.~{\bf
    38}, 9717-9727 (2005).
\bibitem{chefles} A~Chefles, \emph{The Quantum Blackwell Theorem and
    Minimum Error State Discrimination}. ArXiv:0907.0866v4 [quant-ph].

\bibitem{choi} M-D~Choi, \emph{Positive linear maps on
    $C^*$-algebras}. Canad. J. Math. {\bf 24}, 520-529 (1972).

\bibitem{faithful} G~M~D'Ariano and P~Lo~Presti, \emph{Imprinting a
    complete information about a quantum channel on its output
    state}. Phys. Rev. Lett. {\bf 91}, 047902 (2003).
\bibitem{q-black} F~Buscemi, \emph{Comparison of Quantum Statistical
    Models: Equivalent Conditions for
    Sufficiency}. Comm. Math. Phys.~{\bf 310}, 625--647 (2012).
\bibitem{semi-quantum} F~Buscemi, \emph{All Entangled States are
    Nonlocal}. Phys. Rev. Lett.~{\bf 108}, 200401 (2012).
\bibitem{SSW2008} G~Smith, J~A~Smolin, and A~Winter, \textit{The quantum capacity with symmetric side channels}. 	IEEE Trans. Info. Theory \textbf{54}, 9, 4208-4217 (2008).
\bibitem{BO2012} F~G~S~L~Brand\~ao and J~Oppenheim, \textit{The quantum one-time pad in the presence of an eavesdropper}. Phys. Rev. Lett. \textbf{108}, 040504 (2012).
\bibitem{BO2013} F~G~S~L~Brand\~ao and J~Oppenheim, \textit{Public Quantum Communication and Superactivation}. IEEE Trans. Info. Theo. \textbf{59}, 2517 (2013).
\bibitem{footnote2} Note, however, that there exist channels 
which convey the same information to both Bob and Eve, but which cannot be written as $d$-dimensional symmetric channels. An example is given by the 50\% erasure channel mentioned in the introduction, which maps $d$-dimensional inputs into $(d+1)$-dimensional outputs.
\bibitem{SmithSmolin08} Graeme Smith, John A. Smolin, \textit{Additive Extensions of a Quantum Channel}. Proc. of the IEEE Inf. Th. Workshop 2008, pp 368-372.

  \bibitem{schrodinger} E. Schrodinger, Proc. Camb. Phil. Soc. {\bf 31}, 555 (1935).

\bibitem{footnote3} In fact, while $\{\Upsilon^y_A\}_y$ is, in particular, a POVM, $\{\Omega^x_{B_1}\}_x$ in general is not, since $\sum_x\Omega^x\neq\openone$. Nonetheless, they are both complete spanning sets for $\bound(\sH_A)$ and $\bound(\sH_{B_1})$, respectively.

%
%
%
%
%
%
%
%
%
%
%

\end{thebibliography}
%

\appendix

\end{document}